\documentclass{eptcs}

\usepackage{amssymb}
\usepackage{amsmath}
\usepackage{stmaryrd}
\usepackage[mathscr]{eucal}
\usepackage{amsthm}
\usepackage{graphicx}
\usepackage{empheq}
\usepackage{url}
\usepackage{setspace}

\usepackage{tikz} 

\newcommand{\smsm}[1]{\smash{#1}}
\newcommand{\set}[1]{{\,#1\,}}
\newcommand{\mns}{{-}}
\newcommand{\pls}{{+}}
\newcommand{\NonInt}{\Delta}
\newcommand{\wmoda}[1]{(#1)_\NonInt}
\newcommand{\wmod}{w_\NonInt}
\newcommand{\wP}{(P)}

\newcommand{\infinity}{\infty}
\newcommand{\rmdef}{\stackrel{\mbox{\em {\tiny def}}}{=}}
\renewcommand{\S}{\mbox{\large $\rhd\!\!\!\!\!\lhd$}}
\newcommand{\sync}[1]{\raisebox{-1.0ex}{$\;\stackrel{\S}{\scriptscriptstyle
#1}\,$}}

\newcommand{\syncstar}{\smash{\sync{\rule{0pt}{2.8pt}}}\hspace*{-8.7pt}\raisebox{-3.0pt}[0pt][0pt]{$\scriptscriptstyle *$}\hspace*{7pt}}

\newtheorem{definition}{Definition}[section]

\newtheorem{theorem}{Theorem}[section]

\newtheorem{proposition}{Proposition}[section]

\newcommand{\D}{\displaystyle} 
\newcommand{\calC}{\mathcal C}

\newcommand{\calA}{\mathcal A}
\newcommand{\calR}{\mathcal R}
\newcommand{\calF}{\mathcal F}
\newcommand{\calV}{\mathcal V}
\newcommand{\calK}{\mathcal K}
\newcommand{\calN}{\mathcal N}
\newcommand{\calS}{\mathcal S}

\newcommand{\op}{\ \texttt{op}\ }
\newcommand{\opns}{\texttt{op}}
\newcommand{\opnse}{\emph{\texttt{op}}}
\newcommand{\cop}{\texttt{:op}}
\newcommand{\cope}{\emph{\texttt{:op}}}

\newcommand{\dc}{{::}}

\newcommand{\reactant}{\mbox{$\downarrow$}}
\newcommand{\product}{\mbox{$\uparrow$}}
\newcommand{\modifier}{\odot}
\newcommand{\activator}{\oplus}
\newcommand{\inhibitor}{\ominus}
\newcommand{\SE}{\mathit{SE}}
\newcommand{\EI}{\mathit{EI}}

\newcommand{\sst}{\scriptstyle}

\newcommand{\xyrightarrow}[1]{\xrightarrow{\raisebox{-0.9pt}[0pt][0pt]{$\scriptstyle\, #1$\,}}}
\newcommand{\xyRightarrow}[1]{\xRightarrow{\raisebox{-0.0pt}[0pt][0pt]{$\scriptstyle\, #1$\,\,}}}

\newcommand{\Af}{\calA_f}
\newcommand{\As}{\calA_s}

\newcommand{\mm}{{-}}
\newcommand{\mpl}{{+}}

\newenvironment{itemizeless}%
{\begin{list}{$\bullet$}{\setlength{\topsep}{2pt}
               \setlength{\partopsep}{0pt}
               \setlength{\itemsep}{1pt}
               \setlength{\parsep}{1pt}}}
{\end{list}}

\newcounter{countitems}
\newenvironment{enumerateless}%
{\begin{list}{\arabic{countitems}.}{\usecounter{countitems}%
               \setlength{\topsep}{2pt}
               \setlength{\partopsep}{0pt}
               \setlength{\itemsep}{1pt}
               \setlength{\parsep}{1pt}}}
{\end{list}}

\title{A semi-quantitative equivalence for abstracting from fast reactions}

\author{Vashti Galpin$^1$ \qquad\quad Jane Hillston$^{1,2}$
\qquad Federica Ciocchetta
\institute{$^1$ LFCS, School of
Informatics, University of Edinburgh}
\institute{$^2$ CSBE, University of Edinburgh}
\email{Vashti.Galpin@ed.ac.uk, Jane.Hillston@ed.ac.uk}}

\def\titlerunning{A semi-quantitative equivalence for abstracting from fast reactions}
\def\authorrunning{V. Galpin, J. Hillston and F. Ciocchetta}

\begin{document}

\maketitle
\begin{abstract}
Semantic equivalences are used in process algebra to capture the notion of similar behaviour, and this paper proposes a semi-quantitative equivalence for a stochastic process algebra developed for biological modelling. We consider abstracting away from fast reactions as suggested by the Quasi-Steady-State Assumption. We define a fast-slow bisimilarity based on this idea. We also show congruence under an appropriate condition for the cooperation operator of Bio-PEPA.  The condition requires that there is no synchronisation over fast actions, and this distinguishes fast-slow bisimilarity from weak bisimilarity.  We also show congruence for an operator which extends the reactions available for a species.  We characterise models for which it is only necessary to consider the matching of slow transitions and we illustrate the equivalence on two models of competitive inhibition.  \end{abstract}

\section{Introduction}\label{sec:introduction}

One of the features of process algebra is \emph{behavioural or semantic
equivalence} \cite{Miln89,Hill96} which determines if two processes
act the same way. Furthermore, \emph{congruence} is of interest where the
behaviours of two equivalent systems are indistinguishable within
any context, in the sense of combining systems using one or more
operators of the process algebra. Notions of equivalence can be
also useful in \emph{systems biology}. For example, equivalences can be
used to compare the behaviour of two systems or parts of them, and to
show the consistency between different abstractions of the same
system. Furthermore, if congruence holds
then it may be possible to replace a (part of the) system with a smaller
equivalent component without changing the behaviour, and thus reduce
the state space.

We focus on Bio-PEPA \cite{CiocCH:09}, a process algebra 
defined for modelling and analysis of biochemical networks.
Bio-PEPA models have a number of interpretations, including ordinary
differential equations (ODEs) and stochastic simulation \cite{GillG:77}. 
We develop our equivalence in the context of mapping a
Bio-PEPA model to a \emph{finite transition system} (which can be interpreted
as a finite \emph{continuous-time Markov chain} (CTMC)). Such models are
called
\emph{Bio-PEPA models with levels} since 
we assume a
maximum quantity for each species, and we stratify molecule counts or
discretise concentrations into levels, resulting in a finite and
tractable transition system, thus ameliorating the state space explosion
problem.

Here we apply a traditional technique of process algebras, namely
semantic equivalence to Bio-PEPA with levels.  \textit{Isomorphism} and
\textit{strong equivalence} (adapted from PEPA \cite{Hill96}) have
both been defined for Bio-PEPA \cite{CiocCH:09,GalpG:10}, but both
relations are very strong notions of equivalence and not able
to capture biological behaviour of interest.

By contrast, our approach is \emph{semi-quantitative} in that we
consider \emph{relative} rather than actual speeds of reactions,
and use this to determine whether two models have similar behaviour
by abstraction from the faster reactions. The more abstract model
has fewer species and hence fewer parameters. This reduced model
can be parameterised more straightforwardly when
there is limited experimental data.

We define \emph{fast-slow bisimilarity} for Bio-PEPA inspired by biology.
The motivation comes from the \emph{Quasi-Steady-State Assumption} (QSSA)
\cite{SegeS:93} which may reduce systems of ODEs
where there are large differences in reaction rates. This is achieved by
abstracting from fast reactions, by assuming almost no change in
the amount of intermediate products (effectively assuming a steady
state for these products), obtaining fewer reactions and a smaller
system of ODEs.  The rates of reactions in the reduced system are
no longer based on mass action but are defined in more complex
manner which is determined in the derivation of the reduced system of
ODEs. It is then possible to work with the reduced system as a model
thereby requiring fitting of fewer parameters.

As with any semantic equivalence we investigate congruence. Fast-slow
bisimilarity is shown to be a congruence with respect to cooperation
under the condition that there are no fast actions that can occur
between pairs of components. This restriction for congruence is not a
significant limitation as it has been shown that introducing
other fast reactions (beyond those that are considered by QSSA)
involving the species to which QSSA is applied can
drastically reduce the accuracy of the QSSA \cite{SanfSGP:10}.
Fast-slow bisimilarity
is similar in definition to an existing equivalence called \emph{weak
bisimilarity} where the behaviour of silent or invisible $\tau$
actions is abstracted away.  However, the additional condition
needed to show congruence distinguishes them. Using an existing
technique \cite{GomeGVT:08}, we are able to show that for certain
reduced models, it is possible to work with a definition of bisimilarity
which only considers slow reactions.

The rest of the paper is structured as follows.  An introduction
to QSSA follows, then Bio-PEPA is introduced and the definition of
fast-slow equivalence is proposed.  Congruence is proved for
cooperation and the extension operator. Slow bisimulation is defined
and conditions identified for which this is sufficient, followed
by an example to illustrate our equivalence. Finally
related and further work is discussed and concluding remarks are
given.

\section{Quasi-Steady-State Assumption}\label{sec:qssa}

In this section, we consider an existing approach to model reduction
which reduces the number of species to be considered and determines
new reaction rates.  First, consider a set of non-oscillating reactions. We can
identify a set of reactants that are present before the reactions
start, say $\Xi$ and a set of products $\Psi$ that are present once
all reactions have completed.  However, complexes may be created
during the reactions -- these are called \emph{intermediate species}, and
we use $\Upsilon$ for the set of these species. We will also call
the species in $\Phi = \Xi \cup \Psi$ \emph{non-intermediate}.  Note that
$\Upsilon \cap \Phi = \emptyset$ but we cannot assume that $\Xi
\cap \Psi$ is empty since modifiers such as enzymes may appear in
both.

In cellular systems, biochemical reactions can happen on very
different time scales. There can be very frequent reactions
(\textit{fast reactions}) and less frequent reactions (\textit{slow
reactions}). In this case, we can apply the Quasi-Steady-State assumption
\cite{SegeS:93}
which is a time scale separation approach. We discuss other
time scale separation/decomposition techniques in
Section~\ref{sec:relwork}.

If fast reactions lead to the production of intermediate species,
then the instantaneous rates of change of the intermediate species
in the reaction are approximately equal to zero, with respect to
the slow reactions, so they can be viewed as being at steady state.
The pre-steady-state transient period before this happens is much
shorter compared to the time taken for slower reactions, and since
this period is typically of less interest, inaccuracies are less important.

Specifically, we assume species $X_i$,
$i=1,\ldots,n$ with $\Upsilon = \{X_{j_1},\ldots,X_{j_m}\}$ the $m$
intermediate species. The equation\footnote
{Here the concentration of the species $X$ is represented by the
variable $X$.} $dX_{j_k}/dt = f_{j_k}(X)$ (where
$f_{j_k}(X)$ stands for a mathematical expression describing the
dynamics of $X_{j_k}$ in terms of all species) is assumed to be
approximately equal to zero, and therefore $f_{j_k}(X) \approx 0$.
From these equations, it may be possible to derive 
an expression for intermediate species in
terms of other species.  These expressions can replace
the variables representing intermediate species in the rate equations
of other species.  The resulting ODE system has fewer equations and
terms (reactions).  Therefore, the effect of the application of
QSSA is to simplify the model complexity.  This kind of approximation
and resulting reduction is useful when there are many species and
hence many ODEs; when systems of ODEs are stiff (they have widely
different rates) and hence are difficult to solve numerically
\cite{CaoCGP:05a}; and when it is difficult experimentally to obtain
rate parameters for intermediate species.

As an example, consider the reactions $S + E
\xleftrightarrow{R_1,R_{-1}} \SE \xrightarrow{R_2} P+E$ where $S$
is the substrate, $P$ is the product, $E$ the enzyme, $\SE$ the
intermediate substrate-enzyme compound. The double-headed arrow
represents a reversible reaction, with the forward reaction named
$R_1$ and the reverse reaction named $R_{-1}$.  All reactions are
described by mass-action kinetics with rate constants $k_1, \;k_{-1},
\; k_2$. It represents the \textit{Michaelis-Menten mechanism for
enzymatic catalysis} \cite{SegeS:93,MichMM:13}.  The corresponding
ODE system is 

\smallskip

$\begin{array}{rclcrcl}
{dS}/{dt} & = & -k_1 E \cdot S + k_{-1}{\SE} & & 
{dE}/{dt} & = & -k_1 E \cdot S + k_{-1}{\SE} + k_2{\SE} \\
{dP}/{dt} & = & +k_2{\SE} & \text{and} &
{d{\SE}}/{dt} & = & +k_1E \cdot S -k_{-1}{\SE} - k_2{\SE}
\end{array}$

\smallskip

\noindent When the first two reactions are assumed fast and the third slow,
the intermediate species $\SE$ is considered to be at steady-state.
We can say that ${d\SE/dt  \approx 0}$  and from this
we obtain $\SE = {E_T}\cdot S/{(S + K_M)}$,
where $E_T = E + \SE$ is the total enzyme in the system and constant,
and ${K_M = {(k_{-1} + k_2)}/{k_1}}$ (\textit{Michaelis-Menten
constant}). Replacing $\SE$ with the expression above, 
we derive the ODE system
$-{dS}/{dt} = {dP}/{dt} = (k_2 E_T \cdot S)/(S + K_M)$.  
We now have a single reaction $S+E \xrightarrow{} P +E$ that
abstracts the original reactions.  This simplification
is called the \textit{Michaelis-Menten (MM) approximation} and it is
valid under some assumptions, such $S_0 + K_m \gg E_T$ where $S_0$ is the
initial quantity of $S$ \cite{SegeSS:89}.

This approach provides an analogy for the development of our semantic
equivalence which is presented in the next section.  It leads us
to partition reactions into fast and slow, and allows us to abstract
away from the fast reactions. We now introduce the process algebra
to which we will apply these concepts.


\section{Bio-PEPA with levels}\label{sec:bio-pepa}



The syntax of Bio-PEPA with levels \cite{CiocCH:09} is given by the 
grammar below. $S$ defines sequential components which describe the
behaviour of biochemical species, and $P$ defines model
components which combine the species components and from which we
can understand the interactions between species.
\[ S ::= (\alpha, \kappa) \mbox{ \texttt{op} } S \mid S + S \mid C
\qquad \text{\texttt{op}} ::= \reactant \mid \product \mid
\activator \mid \inhibitor \mid \modifier \qquad P::=P \sync{L}  P \mid
S(l) \]
In the term $(\alpha,\kappa) \mbox{ \texttt{op} } S$, 
$\alpha$ is an action or reaction name from $\calA$,
$\kappa \in \{1,2,\ldots\}$ is the
stoichiometric coefficient of the species and the prefix
combinator $\opns$ describes the role of the species in the
reaction. The symbol $\reactant$ is used for a reactant,
$\product$ a product, $\activator$ an activator,
$\inhibitor$ an inhibitor, and $\modifier$ for a generic
modifier.  The operator $+$ provides the choice between two sequential
components and species constants are defined by
$\smash{C \rmdef S}$. 
The process $P \sync{L} Q$ denotes the synchronisation
between two components $P$ and $Q$ and the set $L$ specifies those
reactions on which the components must synchronise. We use $P
\syncstar Q$ to denote the case when all actions shared by $P$ and
$Q$ are synchronised on.  In the model component $S(l)$, the parameter
$l \in \mathbb{N}$ represents the level of molecular count or
concentration.  The set of all Bio-PEPA species components is $\calS$
and the set of all Bio-PEPA model components is $\calC$.

We consider a constrained set of Bio-PEPA models to ensure well-behaved
systems. We require that a species is a choice between reactions
without any repeated actions and that there is only one species
component for a species at model level, as described by the next
definition.

\begin{definition}
\label{wdBs}
A Bio-PEPA sequential component $C$ is \emph{well-defined} if it has the form
\[ C \rmdef (\alpha_1,\kappa_1)\:\opnse_1\: C + \ldots + (\alpha_n,\kappa_n)
\: \opnse_n\: C  \text{\ \ \ written as\ \ \ } 
\textstyle C \rmdef \sum_{i=1}^n (\alpha_i,\kappa_i)\:\opnse_i\:C 
\text{\ \ \ where $\alpha_i \neq \alpha_j$ for $i \neq j$.} \]
A model component $P$ is \emph{well-defined} if it has the form
$P \rmdef C_1(l_1) \sync{\:\:L_1} \ldots \hspace*{-9pt}
\sync{\hspace*{11pt} L_{p-1}}\! C_p(l_p)$
where
each $C_i$ is a well-defined sequential component, the elements of
each $L_j$ appear in $P$ and if $i\neq j$ then $C_i\neq C_j$.
\end{definition}

\noindent Additionally, each model has an associated context collecting
together information such as rates, compartments and parameters, as now
defined.

\begin{definition} A \emph{well-defined Bio-PEPA system} 
$\mathcal{P}$ is a six-tuple
$\langle \mathcal{V},\mathcal{N},\mathcal{K}, \mathcal{F},Comp,P
\rangle$, where $\mathcal{V}$ is the set of compartments,
$\mathcal{N}$ is the set of quantities describing each species,
$\mathcal{K}$ is the set of parameters, $\mathcal{F}$ is
the set of functional rates, $Comp$ is the set of
well-defined sequential components and $P$ is a well-defined model component.
$\mathcal{V}$, $\mathcal{N}$, $\mathcal{K}$, $\mathcal{F}$, $Comp$ 
are called the \emph{context} of $P$.
\end{definition}

For details of the elements of the context and the definition of
well-defined Bio-PEPA system, see \cite{CiocCH:09,CiocCH:08a}. In
this paper, we only consider single compartment models
\footnote{For a presentation
of Bio-PEPA with locations see \cite{CiocCG:09a}.}.
The levels for a species are obtained from information contained
in $\calN$ for that species. More specifically, for each species
$C$, we assume a maximum\footnote{This is reasonable since cells
and other biological compartments have constrained volumes.} molecular
count $M_C$, and a fixed step size $H$ across all species to ensure
conservation of mass during reactions involving multiple species.
The maximum level for a species is determined by $N_C = \lceil M_C/H
\rceil$. Thus, $C$ has levels, $0,\ldots,N_C$, giving $N_C+1$ levels
in total.

\begin{figure*}[t]
\hspace*{-0.2cm}
{\renewcommand{\arraystretch}{0.40}
$\begin{array}{llllll}
\multicolumn{1}{l}{\texttt{\small prefixReac}} &
\D \frac{}{(\alpha,\kappa) \reactant 
S(l) \xrightarrow{(\alpha,[S:\reactant(l,\kappa)])}_c S(l-\kappa)}  
& \multicolumn{3}{l}{\kappa \leq l \leq N_S} \\
\\ \\
\multicolumn{1}{l}{\texttt{\small prefixProd}} &
\D \frac{}{(\alpha,\kappa) \product 
S(l) \xrightarrow{(\alpha,[S:\product(l,\kappa)])}_c S(l+\kappa)} 
& \multicolumn{3}{l}{0 \leq l \leq N_S-\kappa} \\
\\ 
\multicolumn{1}{l}{\texttt{\small prefixMod}} &
\D \frac{}{(\alpha,\kappa) \op
S(l) \xrightarrow{(\alpha,[S:\opns(l,\kappa)])}_c S(l)} 
&
\multicolumn{3}{l}{\!\!\!\!
{\renewcommand{\arraystretch}{0.20}
\begin{tabular}{l}
$\kappa \leq l \leq N_S \text{\ if\ } \opns=\activator$ \\
\\
$0 \leq l \leq N_S \text{\ if\ } \opns\in\{\inhibitor,\modifier\}$
\end{tabular}}}
\\ \\ \\
\multicolumn{1}{l}{\texttt{\small choice1}} &
\D \frac{S_1(l) \xrightarrow{(\alpha,w)}_c S_1'(l')}
{(S_1 + S_2)(l) \xrightarrow{(\alpha,w)}_c S_1'(l')} &
\multicolumn{1}{l}{\texttt{\small \quad choice2}} &
\D \frac{S_2(l) \xrightarrow{(\alpha,w)}_c S_2'(l')}
{(S_1 + S_2)(l) \xrightarrow{(\alpha,w)}_c S_2'(l')} \\
\\
\multicolumn{1}{l}{\texttt{\small coop1}} &
\D \frac{P_1 \xrightarrow{(\alpha,w)}_c P_1'}
{P_1 \sync{L} P_2 \xrightarrow{(\alpha,w)}_c P_1' \sync{L} P_2} 
\quad \alpha \!\not\in\! L &

\multicolumn{1}{l}{\texttt{\small \quad coop2}} &
\D \frac{P_2 \xrightarrow{(\alpha,w)}_c P_2'}
{P_1 \sync{L} P_2 \xrightarrow{(\alpha,w)}_c P_1 \sync{L} P_2'} 
\quad \alpha \!\not\in\! L \\
\\
\multicolumn{1}{l}{\texttt{\small coop3}} &
\multicolumn{1}{l}{\D \frac{P_1 \xrightarrow{(\alpha,w_1)}_c P_1' \quad P_2
\xrightarrow{(\alpha,w_2)} P_2'} 
{P_1 \sync{L} P_2 \xrightarrow{(\alpha,w_1\dc w_2)}_c P_1' \sync{L} P_2'} 
\:\: \alpha \!\in\! L } &
\multicolumn{1}{l}{\texttt{\small \quad constant}} &
\D \frac{S(l) \xrightarrow{(\alpha,[S:\op(l,\kappa)])}_c S'(l')}{C(l)
\xrightarrow{(\alpha,[C:\op(l,\kappa)])}_c S'(l')} \quad C \!\rmdef\! S \\
\\
\\
\multicolumn{1}{l}{\texttt{\small Final}} &
\multicolumn{3}{l}{
\D \frac{P \xrightarrow{(\alpha,w)}_c P'}
   {\langle \calV, \calN, \calK, \calF, Comp, P \rangle
    \xrightarrow{(\alpha,r_\alpha[w,\mathcal{N,K}])}_s
    \langle \calV, \calN, \calK, \calF, Comp, P' \rangle}} \\
\\
\\
\\
\end{array}$}
\caption{Operational semantics of Bio-PEPA}
\label{sos}
\end{figure*}

The operational semantics for Bio-PEPA systems with levels is given
in Figure~\ref{sos} where $N_S$ is the maximum level
for the species $S$. These operational semantics define two labelled
transition systems. The enzyme, inhibitor and general modifier prefixes
are used in reactions that are not modelled as bimolecular reactions with
mass action kinetics, hence the semantics for these prefixes reflect the
fact that the species is not consumed in the reaction and the level
remains the same.

The rules with lowercase letters derive a
relation where transitions are labelled with a reaction name and a
string collecting information about each species involved
in the reaction consisting of the species name, its role in the
reaction, its stoichiometric coefficient for the reaction, and its
current level.  This information is then used in the rule \texttt{Final}
which includes the context of the model to determine the rate of
the reaction and generates the \emph{stochastic relation}\footnote{For
more details on the stochastic relation defined by 
\texttt{Final}, see \cite{CiocCH:09,CiocCH:08a}.  Briefly,
$r_\alpha[w,\calN,\calK] = f_\alpha[w,\calN,\calK]/H \in (0,\infinity)$
where $f_\alpha$ is the functional rate for the reaction $\alpha$
from $\calF$ and $H$ is the step size. $\calN$ provides species
information and $\calK$ provides constants.} from which a CTMC can
be obtained.  In this paper, we focus on the capability relation
as we use this in our definition of equivalence.

\begin{definition}
\emph{A capability label} is defined as $\theta =
(\alpha,w)$ with $\alpha \in \calA$ and the
list $w$ defined by the grammar
$w ::= [\:S\cope\:(l,\kappa)\:] \mid w\:\dc\:w$
where $S \in \calS$, $l \in \mathbb{N}, l \geq 0$ the level, 
$\kappa \in \mathbb{N}, \kappa \geq 1$ the stoichiometric coefficient,
and $::$ is list concatenation.
The set of all such capability labels is $\Theta$.
\end{definition}

\begin{definition}
Given a Bio-PEPA model, the \emph{capability relation} 
${\xrightarrow{}_c} \!\subseteq\! {\calC \!\times\! \Theta \!\times\! \calC}$
is the smallest relation defined by the first nine
rules in Figure~\ref{sos}.
An element of the transition system is written
$P \!\xyrightarrow{\!(\alpha,w)\!}_c\! P'$.
\end{definition}

The string $w$ is defined as a list \cite{CiocCH:09} but the order
of elements is not important so it can be viewed as a multiset.
For well-defined systems, $w$ is a set \cite{GalpG:10} and we will
treat it as such in the sequel.

\begin{definition}
The \emph{derivative set} $ds(P)$ is
the smallest set such that
  $P \in ds(P)$ and
  if $P' \in ds(P)$ and $P'
\xyrightarrow{(\alpha,w)}_{c} P''$ then  $P'' \in ds(P)$. $P' \in ds(P)$
is a \emph{derivative} of $P$.
\end{definition}


This section has defined Bio-PEPA syntax and semantics for Bio-PEPA
systems with levels.  We assume well-defined Bio-PEPA systems with
levels for the remainder of the paper.  Moreover, we assume that
in well-defined Bio-PEPA systems all shared actions are
synchronised over and hence we use the aforementioned notation
$\syncstar$ for cooperation.  The next section considers the
equivalence we develop.

\section{Fast-slow bisimilarity}\label{sec:equiv}

Our basis for developing the equivalence is the QSSA where intermediate
species at steady state can be approximated.  As defined in
Section~\ref{sec:qssa}, we have a set of intermediate species
$\Upsilon$ and non-intermediate species $\Phi$ with $\Upsilon
\cap \Phi = \emptyset$.  When comparing models, we need to ensure
that we exclude intermediate species in the comparison.  Therefore
we define a function that transforms the set $w$ by removing all
intermediate species in $\Upsilon$, and leaving certain species in
$\Delta \subseteq \Phi$. Typically, these species will be those that appear in
transitions for slow reactions in both models with the same role
in both reactions. 
Let $\wmod = \{C\cop\:(l,\kappa)\in w\mid C \in \NonInt \}$ and note that
$\wmoda{w_1\dc w_2} = \wmoda{w_1}\dc\wmoda{w_2}$.

Since QSSA is based on relative reaction rates, we assume that each
reaction can be described as fast or slow, leading to a partition
of the set $\calA$ into the set of slow reactions $\As$ and the set
of fast reactions $\Af$.  For convenience, we introduce new
transitions.

\newpage

\begin{definition} For well-defined Bio-PEPA models $P,P'$.
\smallskip

\hspace*{-0.3cm}
\begin{tabular}{ll}
$\bullet$ \ If $P \xyrightarrow{(\alpha,w)}_c P'$ and 
$\alpha \in \Af$ then $P \twoheadrightarrow P'$. \ \ \ \ 
&
$\bullet$ \ If $P \xyrightarrow{(\alpha,w)}_c P'$ and 
$\alpha \in \As$ then $P \xyrightarrow{\alpha,\wmod} P'$. \\
$\bullet$ \ If $P (\twoheadrightarrow)^* P'$ then $P
\xyRightarrow{\ } P'$. &
$\bullet$ \ If $P (\twoheadrightarrow)^* \xyrightarrow{\alpha,\wmod}
(\twoheadrightarrow)^* P'$ then $P \xyRightarrow{\alpha,\wmod} P'$. 
\end{tabular}
\end{definition}

These new transitions consider the reaction names, the reaction
speeds and certain non-intermediate species involved in the reaction,
hence the equivalence defined will be semi-quantitative in nature.
We now define a new bisimulation based on fast and slow actions.


\begin{definition}
A symmetric relation $\calR$ over $\calC \times \calC$ is a \emph{fast-slow
bisimulation for $\Af$} if $(P,Q) \in \calR$ implies that 
\begin{itemizeless}
\item for all $\alpha \in \As$ whenever $\smash{P
\xyrightarrow{\alpha,\wmod} P'}$ there exists $Q'\!$ with
$\smash{Q \xyRightarrow{\alpha,\wmod} Q'}$ and $(P',Q')\in\!\calR$, and
\item whenever $\smash{P \twoheadrightarrow P'}$ there exists $Q'$ with
$\smash{Q \xyRightarrow{\ } Q'}$ and $(P',Q') \in \calR$
\end{itemizeless}
\end{definition}

Here $\twoheadrightarrow$ plays a similar role to $\xyrightarrow{\tau}$
in Milner's definition of weak bisimulation \cite{Miln89}, and hence
these are similar notions.  Some results for weak bisimulation
may hold for fast-slow bisimulation but there are limits on this,
particularly for proofs based on transition derivations. For example,
when showing congruence where one works with transition
derivations, the difference between $\twoheadrightarrow$ and
$\xyrightarrow{\tau}$ is apparent -- this will be discussed in
more detail later when congruence with respect to the cooperation is
proved.
We can now define a notion of fast-slow bisimilarity with respect to a given
set of fast actions.

\begin{definition}
$P$ and $Q$ are \emph{fast-slow bisimilar for $\Af$} 
($P \thickapprox_{\!\Af} Q$) if there
exists a fast-slow bisimulation for $\Af$,
$\calR$ such that $(P,Q) \in \calR$.
\end{definition}

Then $\thickapprox_{\!\Af}$ is the largest fast-slow bisimulation
for $\Af$.  Now that we have a definition, we wish to show that it
is useful by proving congruence for operators of interest, and
by considering an example.

\section{Congruence}\label{sec:congruence}

When a semantic equivalence captures the notion of same behaviour we
can reason about pairs of systems acting the same. However, if
we show that a semantic equivalence is a congruence with respect to an
operator of the process algebra, then we know that we can build new
systems with equivalent behaviour using that operator. We start by
considering the cooperation operator as it would be useful to know that
we can combine fast-slow bisimilar systems.

To ensure congruence, it is not possible for fast actions to appear
on both sides of the cooperation operator. This makes sense since
this equivalence abstracts away from the details of the fast
reactions, and it is not possible to know if the two models have
abstracted the same reactions. Moreover, recent assessment of
Michaelis-Menten approximation has shown that the QSSA does not
hold if there are other fast reactions (apart from those to which
the QSSA is applied) involving the species that are part of the
Michaelis-Menten reactions \cite{SanfSGP:10}.

\begin{theorem}
If $P_1 \thickapprox_{\!\Af} P_2$ then
$P_1 \syncstar Q \thickapprox_{\!\Af} P_2 \syncstar Q$ and
$Q \syncstar P_1 \thickapprox_{\!\Af} Q \syncstar P_2$ provided
that no action in $\Af$ appears in both $Q$ and $P_1$ or in
both $Q$ and $P_2$.
\end{theorem}
\begin{spacing}{1.10}
\begin{proof}
Let $\calR = \{ (P'_1 \syncstar Q',P'_2 \syncstar Q') \mid
P'_1 \thickapprox_{\!\Af} P'_2 \}$. Consider a transition
$P'_1 \syncstar Q' \xyrightarrow{\alpha,\wmod} R$ which is obtained from 
$\smash{P'_1 \syncstar Q' \xyrightarrow{(\alpha,w)}_c R}$ since $\alpha
\in \As$.
There are three cases and we prove the most complex here.
Assume $P'_1 \xyrightarrow{(\alpha,w_1)}_c P''_1$, $Q'
\xyrightarrow{(\alpha,w_2)}_c Q''$ and $w = w_1\dc w_2$ then $R$ is
$P_1'' \syncstar Q''$. 
Since $P'_1 \thickapprox_{\!\Af} P'_2$, $\smash{P'_2
\xyRightarrow{\alpha,\wmoda{w_1}} P''_2}$ and 
hence $P'_2
\xyrightarrow{(\beta_1,v_1)}_c \cdots
\xyrightarrow{(\beta_n,v_n)}_c
P'_3 \xyrightarrow{(\alpha,w_1)}_c P''_3 \xyrightarrow{(\gamma_1,u_1)}_c
\cdots \xyrightarrow{(\gamma_{m},u_{m})}_c P''_2$ for the actions
$\beta_1,\ldots,\beta_n,\gamma_1,\ldots,\gamma_{m} \in \Af$.
From this, we can derive 
$P'_2 \syncstar Q' \xyrightarrow{(\beta_1,v_1)}_c \cdots
\xyrightarrow{(\beta_n,v_n)}_c P'_3 \syncstar Q'
\xyrightarrow{(\alpha,w_1\dc w_2)}_c P''_3 \syncstar Q''
\xyrightarrow{(\gamma_1,u_1)}_c\: \:\cdots\:
\xyrightarrow{(\gamma_{m},u_{m})}_c\: P''_2 \syncstar Q''$, hence
we have $\smash{P'_2 \syncstar Q' \xyRightarrow{\alpha,
\wmoda{w_1\dc w_2}} P''_2 \syncstar Q''}$ as required with $(P''_1 \syncstar
Q'',P''_2 \syncstar Q'') \in \calR$.

Next, consider a transition $\smsm{P'_1 \syncstar Q' \twoheadrightarrow{\
} R}$. Hence there exists $\beta \in \Af$ such that $\smsm{P'_1
\syncstar Q' \xyrightarrow{(\beta,v)}_c R}$. There are two cases
since cooperation over an action in $\Af$ is excluded by the
condition.  We show the case where $Q' \xyrightarrow{(\beta,v)}_c
Q''$ and $R$ is $P_1' \syncstar Q''$. Then the transition $P'_2
\syncstar Q' \xyrightarrow{(\beta,v)}_c P'_2 \syncstar Q''$ can be
derived and hence $P'_2 \syncstar Q' \xyRightarrow{\ } P'_2 \syncstar
Q''$ with $(P'_1 \syncstar Q'',P'_2 \syncstar Q'') \in \calR$ as
required.
\end{proof}
\end{spacing}

To see why the sharing of fast actions must be prohibited, consider
$\smash{S_1 \rmdef (\alpha,2)\uparrow S_1 + (\gamma,2)\downarrow
S_1}$ and $\smash{S_2 \rmdef (\beta,2)\uparrow S_2 + (\gamma,2)\downarrow
S_2}$. Clearly $S_1(0) \approx_{\{\alpha,\beta\}} S_2(0)$. Considering
a third species $\smash{S \rmdef (\alpha,1) \downarrow S}$, it is
not the case that ${S_1(0) \syncstar S(1)} \approx_{\{\alpha,\beta\}}
{S_2(0) \syncstar S(1)}$ since these systems have very different
behaviours because there are no $\gamma$ reactions in the first systems
and there are repeating $\gamma$ reactions in the second. To prevent
this difference, $S$ could be modified to require that it perform
all fast reactions giving $\smash{S \rmdef (\alpha,1) \downarrow S
+ (\beta,1) \downarrow S}$ but it is not clear how to generalise
this beyond species. This counter-example for the condition in the
theorem demonstrates that the definition of fast-slow bisimulation
does differ from that of weak bisimulation. Here we abstract away
from fast reaction names on a transition, whereas with weak
bisimulation, transitions with the named action $\tau$ are treated
abstractly.

\subsection{The species extension operator}

Since the focus is well-defined Bio-PEPA models, it is not useful to
consider the operators for sequential agents individually. However, we
do sometimes want to extend species' ability to be involved in
reactions, and we now define an operator that permits this to happen
and show that we have congruence for this operator under certain
conditions.

When we build the models of biological systems, we may have no
knowledge of what other reactions the species could be involved in.
For example, considering the product of some sequence of reactions,
we can imagine various scenarios in which we would want
the product to be able to react with new species added by cooperation.
It is difficult to do this at model level but we can consider it
at species level by defining an appropriate extension operator
\cite{GalpG:10}.

\begin{definition}
Given two well-defined species $A$ and $B$ such that the reaction names of
$A$ are disjoint from those of $B$, define $A\{B\}$, the 
\emph{extension of} $A$ \emph{by} $B$ as

$\begin{array}{c}
\hspace*{-0.7cm}
\D
A\{B\} \!=\! \sum_{i=1}^n (\alpha_i,\kappa_i) \opnse_{i\:} A\{B\} +
           \sum_{j=1}^m (\beta_j,\lambda_j) \opnse_{\!j\,} A\{B\} \quad 
           \text{where} \quad
A \!=\! \sum_{i=1}^n (\alpha_i,\kappa_i) \opnse_{i\;} A, \quad 
   B \!=\! \sum_{j=1}^m (\beta_j,\lambda_j) \opnse_{\!j\,} B. 
\end{array}$
\end{definition}

This permits $A$ to take on additional reaction capabilities,
specifically those of $B$. $A\{B\}$ is well-defined since there are
no repeated reaction names and because $A$ and $B$ are well-defined.
Note that $A\{B\}$ and $B\{A\}$ are
isomorphic since their transition systems are structurally identical
with matching actions.
The next theorem shows how species can be augmented with ways to
participate in new reactions in a way which preserves fast-slow
bisimulation. 

\begin{theorem}
Let $C_1(l) \thickapprox_{\!\Af}\! C_2(l)$ for
sequential Bio-PEPA components $C_1$ and $C_2$ for all $0 \leq l \leq
N_{C_1}=N_{C_2}$ and let $C$ be a
well-defined species with reaction names disjoint from those in $C_1$
and $C_2$.
Then $C_1\{C\}(l) \thickapprox_{\!\Af}\! C_2\{C\}(l)$
and $C\{C_1\}(l) \thickapprox_{\!\Af}\! C\{C_2\}(l)$.
\end{theorem}
\begin{spacing}{1.00}
\begin{proof}
Consider a transition $\smsm{C_1\{C\}(l) \xyrightarrow{\alpha,\wmod} C_1\{C\}(l')}$ for
$\alpha \in \As$ then $\smsm{C_1\{C\}(l) \xyrightarrow{(\alpha,w)}_c
\!C_1\{C\}(l')}$. If $\alpha$ appears in $C$, then $\smsm{C_2\{C\}(l)
\!\xyrightarrow{(\alpha,w)}_c\!C_2\{C\}(l')}$ and $\smsm{C_2\{C\}(l)
\xyRightarrow{\alpha,\wmod}\!C_2\{C\}(l')}$.

If $\alpha$ appears in $C_1$ and $C_2$ then $\smsm{C_1(l)
\xyrightarrow{\alpha,\wmod} C_1(l')}$ and since $C_1(l)
\thickapprox_{\!\Af} C_2(l)$, then $\smsm{C_2(l)
\xyRightarrow{\alpha,\wmod} C_2(l'')}$. This then gives
$\smsm{C_2(l)\xyrightarrow{(\beta_1,v_1)}_c \:\cdots\:
\xyrightarrow{(\beta_n,v_n)}_c C_2(m) \, \xyrightarrow{(\alpha,w)}_c
C_2(m') \xyrightarrow{(\gamma_1,u_1)}_c}$ $\smsm{\cdots
\xyrightarrow{(\gamma_{m},u_{m})}_c C_2(l'')}$ for the actions
$\beta_1,\ldots,\beta_n,\gamma_1,\ldots,\gamma_{m} \in \Af$.
$C_2\{C\}(l)$ can perform the same actions hence 
$\smsm{C_2\{C\}(l)
\xyRightarrow{\alpha} C_2\{C\}(l'')}$.

Next, consider $\smsm{C_1\{C\}(l) \twoheadrightarrow
C_1\{C\}(l')}$.  There exists $\beta \in \Af$ such that 
$\smsm{C_1\{C\}(l) \xyrightarrow{(\beta,v)}_c C_1\{C\}(l')}$. If $\beta$ appears
in $C$ then as above, there is a matching
transition in $C_2\{C\}(l)$. If $\beta$ appears in $C_1$ since
$C_1(l) \thickapprox_{\!\Af} C_2(l)$, then $\smsm{C_2(l)
\xyRightarrow{\ } C_2(l'')}$ which gives $\smsm{C_2(l)
\xyrightarrow{(\beta_1,v_1)}_c \cdots
\xyrightarrow{(\beta_n,v_n)}_c C_2(l'')}$ for
$\beta_1,\ldots,\beta_n \!\in\! \Af$. $C_2\{C\}(l)$ can perform these
actions,
hence $\smsm{C_2\{C\}(l)
\xyRightarrow{\ } C_2\{C\}(l'')}$ as required. 
\end{proof}
\end{spacing}

Let $C_1 \rmdef (\alpha,1) \uparrow C_1$ and  
$C_2 \rmdef (\alpha,1) \uparrow C_2 + (\beta,1) \oplus C_2$ with shared
maximum level, then $C_1(l) \thickapprox_{\{\beta\}} C_2(l)$. For any
sequential component $C$ with different reaction names 
to $\alpha$ and
$\beta$ and the same maximum level, the congruence result applies.


\section{Slow bisimilarity}

Checking for fast-slow bisimilarity requires that all reactions
must be considered.  We now define an equivalence over just the
slow reactions. If we can identify conditions under which this
equivalence implies fast-slow bisimilarity, then whenever we have
models that satisfy those conditions, we need only check the slow
reactions to prove that a relation is a fast-slow bisimulation.
This section introduces such an equivalence and conditions. 
In the next section when we consider competitive inhibition, we
will illustrate how compact our proofs are.
First we define the new equivalence. As before, it is assumed that $\As$
and $\Af$ partition $\calA$.

\begin{definition}
A symmetric relation $\calR$ over $\calC \times \calC$ is a \emph{slow
bisimulation for $\As$} if $(P,Q) \in \calR$ implies that
for all $\alpha \in \As$ whenever
\begin{itemize}
\item $\smash{P
\xyrightarrow{\alpha,\wmod} P'}$ there exists $Q'\!$ with
$\smash{Q \xyRightarrow{\alpha,\wmod} Q'}$ and $(P',Q')\in\!\calR$
\end{itemize}
$P$ and $Q$ are \emph{slow bisimilar for $\As$} if there
exists a slow bisimulation for $\As$, $\calR$ such that
$(P,Q) \in \calR$.
\end{definition}

We now consider when this can be applied using a technique that allow
variables to be classified by what reactions affect them.

\subsection{Variable classification}

We present an existing technique that allows the identification of
slow variables (those that are only affected by slow reactions) and
fast variables (those that are affected by fast and slow reactions)
\cite{GomeGVT:08}.  Note that variables are not the same as species
since a variable can either be an individual species or a linear
combinations of species.

A set of reactions can be expressed as a stoichiometry matrix
$\mathbf{S}$ which has $m$ columns, one for each reaction and $n$
rows, one for each species. $S_{i,j}$ describes the stoichiometry
of species $X_i$ with respect to reaction $R_j$. 


A stoichiometry matrix can be transformed into a matrix of the same
size with the form described in Figure~\ref{fig:matrix} \cite{GomeGVT:08}.
As mentioned above, the variables that are associated with the rows
of $\mathbf{Q}$ are linear combinations of the original species
variables and hence when given values for the new variables, it is
possible to establish values for the species.  

The top row of submatrices in the figure are zeroes since these
represent conserved variables and reactions do not affect them.
$\mathbf{Q}_{ss}$ has size $n_s \times m_s$ and captures the
stoichiometry of slow reactions for slow variables. The other
submatrix in its row is zero since slow variables are not affected
by fast reactions. The last row of $\mathbf{Q}$ consists of an $n_f
\times m_s$ matrix and an $n_f \times m_f$ matrix describing the
stoichiometry of slow and fast reactions with respect to fast
variables.

For a given ordering of reactions and variables, $\mathbf{Q}$ is
unique. However, for different orderings, an equally valid but
different $\mathbf{Q}$ may be obtained. If there are no slow variables
then this technique cannot be used.

For reasons of space, it is not possible to fully describe the
matrix transformation defined in \cite{GomeGVT:08}. The idea is
based around invariants. These are variables whose values are not
changed by the dynamics of the model.  First, invariants (conserved
variables) of the
model are identified.  Then, by removing the slow reactions from
the reaction equations, it is possible to find slow variables
(invariants when slow reactions are removed), if any. Then sufficient
fast variables must be identified so that there are $n$ variables
in total. Each new variable must be linearly independent of the
other new variables, and the new variables are linear combinations
of the original species variables.  This process is illustrated in
the example section.

\begin{figure}
\begin{tabular}{lp{13cm}}
\\
\\

$\mathbf{Q} =
{\renewcommand{\arraystretch}{1.6}
\begin{bmatrix}
\;\mathbf{0}      &\mathbf{0}\: \\
\;\mathbf{Q}_{ss} & \mathbf{0}\: \\
\;\mathbf{Q}_{fs} & \mathbf{Q}_{ff}\: \\
\end{bmatrix}}$
\hspace*{-0.7cm}

& \vspace*{-1.05cm}

\small{\begin{itemizeless}
\item
Columns represent reactions $R_1,\ldots
R_{m_s},R'_{1}\ldots R'_{m_f}$ where the $R_i$ are $m_s$ slow
reactions, the $R'_j$ are $m_f$ fast reactions, and $m_s\pls  m_f =
m$
\item Rows represent variables
$X_1,\ldots,X_{n_c},X'_1, \ldots X'_{n_s},X''_1, \ldots X''_{n_f}$ where
the $X_i$ are conserved variables, the $X'_j$ are slow
variables, the $X''_k$  are fast variables and $n_c\pls n_s\pls n_f = n$
\end{itemizeless}}
\\
\end{tabular}
\caption{Matrix transformation \cite{GomeGVT:08}.}
\label{fig:matrix}
\end{figure}

\vspace*{2mm}

\subsection{Application to Bio-PEPA}

Given a Bio-PEPA model, we can construct its stoichiometry matrix
from the species component definitions. Using the technique described
above we can identify invariants, slow and fast variables\footnote{These
invariants are related to P-invariants in Petri nets \cite{ClarCGGGH:12a}.
Invariants can be determined automatically by the Bio-PEPA Plug-in
\cite{ClarCGGK:10a}. Since the Bio-PEPA Plug-in allows reactions to be
removed when inferring invariants, slow variables can also be found
automatically by removing slow reactions. See also
\url{www.biopepa.org}.}.

Note that a well-defined Bio-PEPA model only differs from its
derivatives in terms of the levels of each species, hence models
can be represented as vectors where each element represents the
level of a species.  See Figure~\ref{fig:tss} for an example.

A model's transition system, the capability relation, is then 
defined over states that are given in vector form $(v_1,\ldots,v_p)$
for $p$ species.  These states can be transformed to
$(s_1,\ldots,s_{n_s},f_1,\ldots,f_{n_f})$ where $n_s\pls n_f = p\mns
n_c$, producing a new transition system where the transitions are
unchanged and the states are defined with respect to the values of
the new variables, specifically the slow and fast variables.
Conserved variable values are not included in the new states since
their values are fixed.  The new states contain the same information
as the original states, and they therefore stay unique. It
is not possible for two states in the original transition system
to collapse into one state in the new transition system.  We can
conclude that the transition systems are isomorphic, meaning that
there is a bijection between states, and transitions are preserved
with the same labels.

We now identify conditions that allow us to show when a slow
bisimulation is a fast-slow bisimulation using variable classification.
We restrict ourselves to the case of equivalence between a model
which has conserved, slow and fast variables and a model that has
conserved variables, slow variables which are the same as those in
the first model, and no fast variables.  We also require that slow
variables are individual species. Extending the result to more
general cases is further work.


\newpage

\begin{proposition}\label{prop:slowfast}
Consider two Bio-PEPA models $P_i$ for $i=1,2$
where
$\Delta_i$ contains exactly the slow species of the model, such that
$\Delta_1$ and $\Delta_2$ have the same species.
Let $\calR$ be a relation over Bio-PEPA models such that for all
$\bigl((s_1,\ldots,s_{n_s},f_1,\ldots,f_{n_f}),
(s'_1,\ldots,s'_{n_s})\bigr)\in \calR$, the $s_i$ and $s'_i$
are values for all slow variables and the $f_i$ are the values for fast
variables
If $s_i=s'_i$ for $i \in \{1,\ldots n_s\}$ and $\calR$ is a slow
bisimulation for $\As$, then $\calR$ is a fast-slow bisimulation
for $\Af$.
\end{proposition}
\begin{proof}
Let $\calR$ be a slow bisimulation with the required condition. Hence we
need to consider fast actions only.
Let $\bigl((s_1,\ldots,s_{n_s},f_1,\ldots,f_{n_f}),
(s_1,\ldots,s_{n_s})\bigr) \in \calR$ and consider
the fast transition such that
$(s_1,\ldots,s_{n_s},f_1,\ldots,f_{n_f}) \!\twoheadrightarrow \!
(s_1,\ldots,s_{n_s},f'_1,\ldots,f'_{n_f})$. We know that
$(s_1,\ldots,s_{n_s}) \!\xyRightarrow{\ }\!
(s_1,\ldots,s_{n_s})$ and also that
$\bigl((s_1,\ldots,s_{n_s},f'_1,\ldots,f'_{n_f}),
(s_1,\ldots,s_{n_s})\bigr) \in \calR$. There are no fast actions
from $(s_1,\ldots,s_{n_s})$ to consider.
\end{proof}
\noindent Given two Bio-PEPA models, the general technique can be summarised as follows. 
\begin{enumerateless}
\item Classify the variables in each model. Check that one model only
has slow variables and that the slow variables are species and the same between
models. If not, try different variable orderings.
\item Transform the transition systems of both models as described
above.
\item Define a relation over the transformed transition systems of the
two models where slow variables
have the same value in each pair in the relation.
\item Check that this relation is a slow bisimulation,
and use Proposition~\ref{prop:slowfast} to show that that it is a
fast-slow bisimulation.
\item Since the transformation has provided an isomorphic transition
system, the original models are fast-slow bisimilar.
\end{enumerateless}




\section{Competitive inhibition}\label{sec:example}


We now consider an example where there are significantly different
rates and hence a suitable test case for
fast-slow bisimulation. 
It is an example of competitive inhibition \cite{SegeS:93} where
an inhibitor is introduced, giving the reactions
$S + \EI \: \xleftrightarrow{\ \ \ } S + E + I \: \xleftrightarrow{\ \
\ } \SE + I
\: \rightarrow \: P + E + I$.
Here, the first reversible reaction describes how the enzyme and
inhibitor can bind together to form a compound. The second reversible
reaction shows how the
substrate and enzyme bind together to form a compound from which
the product can be obtained. The binding of the inhibitor and enzyme
competes with the binding of the enzyme and substrate since when
the enzyme is bound to the inhibitor it is not available for the
reaction with the substrate and hence reduces the amount of product
that can be produced. Then $\Xi = \{S, E, I\}$, $\Psi = \{P,E,I\}$
and $\Upsilon = \{\EI, \SE\}$ since $\EI$ and $\SE$ are the intermediate
species (as defined in Section~\ref{sec:qssa}) created by these
reactions. Because of the explicit representations of the
inhibitor and enzyme, and their associated intermediates, we choose
to model the basic bimolecular reactions with mass actions kinetics.

These reactions can be expressed in Bio-PEPA as follows, where
$\alpha_1$ and $\alpha_{-1}$ are the reactions involving enzyme and
inhibitor, $\beta_1$ and $\beta_{-1}$ are the reactions involving
substrate and enzyme, and $\gamma$ is the reaction that produces
the product.

\smallskip

\hspace*{-0.2cm}
{\renewcommand{\arraystretch}{1.35}
$\begin{array}{rclrclrcl}
\hspace*{-0.2cm}
S & \rmdef & (\beta_1,1)\set{\downarrow} S + (\beta_{-1},1)\set{\uparrow} S &
P & \rmdef & (\gamma,1)\set{\uparrow} P \quad \quad &
I & \!\!\!\rmdef\!\!\! & (\alpha_1,1)\set{\uparrow} I +
(\alpha_{-1},1)\set{\downarrow} I \\
\hspace*{-0.2cm}
\EI & \rmdef & (\alpha_1,1)\set{\downarrow} \EI +
(\alpha_{-1},1)\set{\uparrow} \EI
\quad\quad &
\SE & \rmdef & \multicolumn{4}{l}{(\beta_1,1)\set{\uparrow} \SE + 
(\beta_{-1},1)\set{\downarrow} \SE + (\gamma,1)\set{\downarrow} \SE} \\
\hspace*{-0.2cm}
E & \rmdef & \multicolumn{7}{l}{(\alpha_1,1)\set{\uparrow} E + 
(\alpha_{-1},1)\set{\downarrow} E + (\beta_1,1)\set{\downarrow} E + 
  (\beta_{-1},1)\set{\uparrow} E + (\gamma,1)\set{\uparrow} E} \\
\hspace*{-0.2cm}
Sys & \rmdef & \multicolumn{7}{l}{S(l_S) \syncstar E(l_E) 
\syncstar I(l_I) \syncstar P(l_P) \syncstar  EI(l_{\EI}) \syncstar
SE(l_{\SE})} \\
\end{array}$}

\noindent Here, based on biological understanding, we set 
$\{\alpha_1,\alpha_{-1},\beta_1,\beta_{-1}\} = \Af$, namely
that these are the fast reactions, and that $\gamma \in \As$.

We wish to show that this is fast-slow bisimilar to the simpler
Bio-PEPA model defined as follows. In this model, only a single
reaction is modelled and this reaction has a rate which takes
into account the amount of enzyme and inhibitor present.  This
reaction is considered to be a slow reaction.  Since this reaction
is not based on mass actions kinetics, the inhibitor and enzyme
prefix operators are used.  The reaction is named $\gamma$ as it
produces $P$, as in the previous model.  The use of primes on species
and model names is a syntactic convenience to distinguish different
species and model components. However, later on we will view $P$
and $P'$ as the same when we consider $\wmod$.

\smallskip

{\renewcommand{\arraystretch}{1.35}
$\begin{array}{rclrclrclrcl}
S' & \rmdef & (\gamma,1)\set{\downarrow} S' \quad \quad &
E' & \rmdef & (\gamma,1)\set{\oplus} E' \quad \quad &
I' & \rmdef & (\gamma,1)\set{\ominus} I' \quad \quad &
P' & \rmdef & (\gamma,1)\set{\uparrow} P' \\
&  &  & Sys' & \rmdef & \multicolumn{7}{l}{S'(l_{S'}) \syncstar E'(l_{E'}) 
\syncstar I'(l_{I'}) \syncstar P'(l_{P'})} \\
\end{array}$}

\smallskip

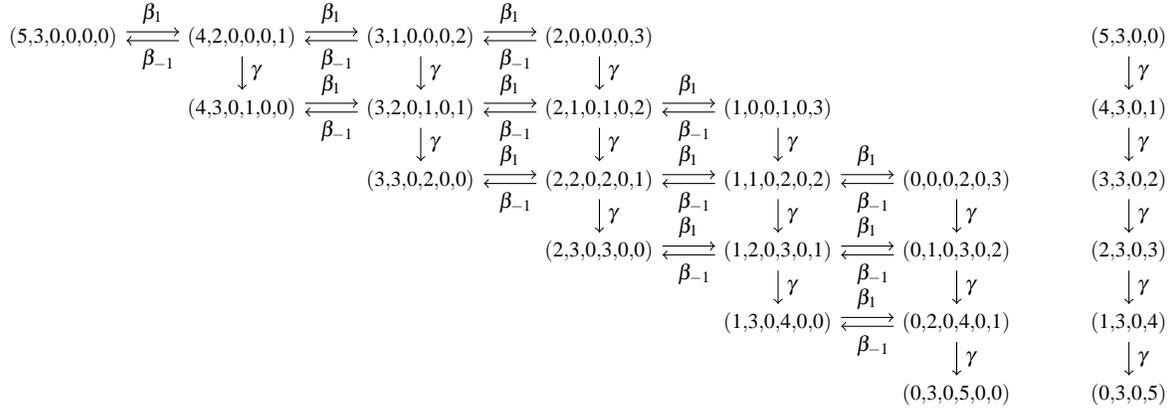
\begin{figure*}
\begin{tabular}{cccc}
\begin{tikzpicture}[xscale=0.95,yscale=0.95,inner sep=3pt] 

\node (t5300) at (0,0) {$\sst (5,3,0,0,0,0)$};
\node (t4210) at (2.5,0) {$\sst (4,2,0,0,0,1)$};
\node (t3120) at (5.0,0) {$\sst (3,1,0,0,0,2)$};
\node (t2030) at (7.5,0) {$\sst (2,0,0,0,0,3)$};

\node (t4301) at (2.5,-1) {$\sst (4,3,0,1,0,0)$};
\node (t3211) at (5.0,-1) {$\sst (3,2,0,1,0,1)$};
\node (t2121) at (7.5,-1) {$\sst (2,1,0,1,0,2)$};
\node (t1031) at (10.0,-1) {$\sst (1,0,0,1,0,3)$};

\node (t3302) at (5.0,-2) {$\sst (3,3,0,2,0,0)$};
\node (t2212) at (7.5,-2) {$\sst (2,2,0,2,0,1)$};
\node (t1122) at (10.0,-2) {$\sst (1,1,0,2,0,2)$};
\node (t0032) at (12.5,-2) {$\sst (0,0,0,2,0,3)$};

\node (t2303) at (7.5,-3) {$\sst (2,3,0,3,0,0)$};
\node (t1213) at (10.0,-3) {$\sst (1,2,0,3,0,1)$};
\node (t0123) at (12.5,-3) {$\sst (0,1,0,3,0,2)$};

\node (t1304) at (10.0,-4) {$\sst (1,3,0,4,0,0)$};
\node (t0214) at (12.5,-4) {$\sst (0,2,0,4,0,1)$};

\node (t0305) at (12.5,-5) {$\sst (0,3,0,5,0,0)$};

\path[->,draw] (node cs:name=t5300, angle=5) to 
     (node cs:name=t4210, angle=175) node[midway,above] {$\sst \beta_1$};
\path[->,draw] (node cs:name=t4210, angle=5) to 
     (node cs:name=t3120, angle=175) node[midway,above] {$\sst \beta_1$};
\path[->,draw] (node cs:name=t3120, angle=5) to 
     (node cs:name=t2030, angle=175) node[midway,above] {$\sst \beta_1$};

\path[->,draw] (node cs:name=t4210, angle=183) to 
(node cs:name=t5300, angle=357) node[midway,below] {$\sst\:\:\:\:\beta_{-1}$};
\path[->,draw] (node cs:name=t3120, angle=183) to 
(node cs:name=t4210, angle=357) node[midway,below] {$\sst\:\:\:\:\beta_{-1}$};
\path[->,draw] (node cs:name=t2030, angle=183) to 
(node cs:name=t3120, angle=357) node[midway,below] {$\sst\:\:\:\:\beta_{-1}$};

\path[->,draw] (node cs:name=t4301, angle=5) to 
     (node cs:name=t3211, angle=175) node[midway,above] {$\sst \beta_1$};
\path[->,draw] (node cs:name=t3211, angle=5) to 
     (node cs:name=t2121, angle=175) node[midway,above] {$\sst \beta_1$};
\path[->,draw] (node cs:name=t2121, angle=5) to 
     (node cs:name=t1031, angle=175) node[midway,above] {$\sst \beta_1$};

\path[->,draw] (node cs:name=t3211, angle=183) to 
(node cs:name=t4301, angle=357) node[midway,below] {$\sst\:\:\:\:\beta_{-1}$};
\path[->,draw] (node cs:name=t2121, angle=183) to 
(node cs:name=t3211, angle=357) node[midway,below] {$\sst\:\:\:\:\beta_{-1}$};
\path[->,draw] (node cs:name=t1031, angle=183) to 
(node cs:name=t2121, angle=357) node[midway,below] {$\sst\:\:\:\:\beta_{-1}$};

\path[->,draw] (node cs:name=t3302, angle=5) to 
     (node cs:name=t2212, angle=175) node[midway,above] {$\sst \beta_1$};
\path[->,draw] (node cs:name=t2212, angle=5) to 
     (node cs:name=t1122, angle=175) node[midway,above] {$\sst \beta_1$};
\path[->,draw] (node cs:name=t1122, angle=5) to 
     (node cs:name=t0032, angle=175) node[midway,above] {$\sst \beta_1$};

\path[->,draw] (node cs:name=t2212, angle=183) to 
(node cs:name=t3302, angle=357) node[midway,below] {$\sst\:\:\:\:\beta_{-1}$};
\path[->,draw] (node cs:name=t1122, angle=183) to 
(node cs:name=t2212, angle=357) node[midway,below] {$\sst\:\:\:\:\beta_{-1}$};
\path[->,draw] (node cs:name=t0032, angle=183) to 
(node cs:name=t1122, angle=357) node[midway,below] {$\sst\:\:\:\:\beta_{-1}$};

\path[->,draw] (node cs:name=t2303, angle=5) to 
     (node cs:name=t1213, angle=175) node[midway,above] {$\sst \beta_1$};
\path[->,draw] (node cs:name=t1213, angle=5) to 
     (node cs:name=t0123, angle=175) node[midway,above] {$\sst \beta_1$};

\path[->,draw] (node cs:name=t1213, angle=183) to 
(node cs:name=t2303, angle=357) node[midway,below] {$\sst\:\:\:\:\beta_{-1}$};
\path[->,draw] (node cs:name=t0123, angle=183) to 
(node cs:name=t1213, angle=357) node[midway,below] {$\sst\:\:\:\:\beta_{-1}$};

\path[->,draw] (node cs:name=t1304, angle=5) to 
     (node cs:name=t0214, angle=175) node[midway,above] {$\sst \beta_1$};

\path[->,draw] (node cs:name=t0214, angle=183) to 
 (node cs:name=t1304, angle=357) node[midway,below] {$\sst\:\:\:\:\beta_{-1}$};

\path[->,draw] (t4210) to (t4301) node[midway,right] {$\sst \gamma$};

\path[->,draw] (t3120) to (t3211) node[midway,right] {$\sst \gamma$};
\path[->,draw] (t3211) to (t3302) node[midway,right] {$\sst \gamma$};

\path[->,draw] (t2030) to (t2121) node[midway,right] {$\sst \gamma$};
\path[->,draw] (t2121) to (t2212) node[midway,right] {$\sst \gamma$};
\path[->,draw] (t2212) to (t2303) node[midway,right] {$\sst \gamma$};

\path[->,draw] (t1031) to (t1122) node[midway,right] {$\sst \gamma$};
\path[->,draw] (t1122) to (t1213) node[midway,right] {$\sst \gamma$};
\path[->,draw] (t1213) to (t1304) node[midway,right] {$\sst \gamma$};

\path[->,draw] (t0032) to (t0123) node[midway,right] {$\sst \gamma$};
\path[->,draw] (t0123) to (t0214) node[midway,right] {$\sst \gamma$};
\path[->,draw] (t0214) to (t0305) node[midway,right] {$\sst \gamma$};

\end{tikzpicture} 

& & 

\begin{tikzpicture}[xscale=0.95,yscale=0.95,inner sep=3pt] 

\node (t530) at (0,0) {$\sst (5,3,0,0)$};

\node (t431) at (0,-1) {$\sst (4,3,0,1)$};

\node (t332) at (0,-2) {$\sst (3,3,0,2)$};

\node (t233) at (0,-3) {$\sst (2,3,0,3)$};

\node (t134) at (0,-4) {$\sst (1,3,0,4)$};

\node (t035) at (0,-5) {$\sst (0,3,0,5)$};

\path[->,draw] (node cs:name=t530, angle=270) to 
      (node cs:name=t431, angle=90) node[midway,right] {$\sst \gamma$};
\path[->,draw] (node cs:name=t431, angle=270) to 
      (node cs:name=t332, angle=90) node[midway,right] {$\sst \gamma$};
\path[->,draw] (node cs:name=t332, angle=270) to 
      (node cs:name=t233, angle=90) node[midway,right] {$\sst \gamma$};
\path[->,draw] (node cs:name=t233, angle=270) to 
      (node cs:name=t134, angle=90) node[midway,right] {$\sst \gamma$};
\path[->,draw] (node cs:name=t134, angle=270) to 
      (node cs:name=t035, angle=90) node[midway,right] {$\sst \gamma$};

\end{tikzpicture} 
\end{tabular}

\caption{Transition system for $Sys$ and $Sys'$ for $n=5$ and $m=3$ with no
inhibitor present (only reaction names appear on transitions).}\label{fig:tss}
\end{figure*}
From the equations, it is clear that for a starting level of the
substrate $S$ (or $S'$) given by $l_S=n$ it is not possible to reach
more than $n$ levels of $P$ (alternatively $P'$) if the starting level
of $P$ is set to $l_P=0$.  This agrees with the biological understanding
that these reactions represent a transformation of $S$ to $P$ through a
number of steps.

We assume neither of the intermediates nor the product are
present at the start in the more complex model as is standard
\cite{SegeSS:89}.
Hence, for a starting level of $m$ of the enzyme, it is not possible
to have more than $m$ levels of the substrate-enzyme compound, and
for a starting level of $m$ of the enzyme and $p$ of the inhibitor,
it is not possible to have more than $\min\{m,p\}$ levels of the
enzyme-inhibitor compound. 

As mentioned above, a well-defined Bio-PEPA model only differs from
its derivatives in the levels of the species and models and derivatives
can be expressed in numeric vector form.  For example, for $Sys$
the vector $(2,0,3,1,0,4)$ describes the model with 2 levels of
$S$, none of $E$, 3 of $I$, 1 of the product $P$, none of the
compound $\EI$ and 4 of the compound $\SE$.

\vspace*{-2.7mm}

\subsection{Constructing the bisimulation}

Under the initial species levels described above, there are four
cases of interest: only substrate present, substrate
and enzyme present, substrate and inhibitor present, and substrate,
enzyme and inhibitor present. 
These can be
considered in one relation over the two models with starting vectors
$(n,m,p,0,0,0)$ (using the ordering $(l_{S},l_{E},l_{I},l_{P},l_{\EI},l_{\SE})$)
and $(n,m,p,0)$ (using the ordering
$(l_{S'},l_{E'},l_{I'},l_{P'})$) with $n > 0$ and $m, p \geq 0$.
Figure~\ref{fig:tss} illustrates the case when $n=5$, $m=3$ and $p=0$.
This case with no inhibitor represents an instance of the standard
Michaelis-Menten mechanism \cite{SegeS:93} as discussed in
Section~\ref{sec:qssa}.

\newpage
To define fast-slow bisimulation, we must determine which
non-intermediate species are in the set $\NonInt$. 
The label on the transition of a $\gamma$ reaction in $Sys$ is
$(\gamma,w)$ where $w=
\{ P{:}\set{\uparrow}(1,i_1), SE{:}\set{\downarrow}(1,i_2),
E{:}\set{\uparrow}(1,i_3) \}$. For a $\gamma$ in $Sys'$, the set is
$\{ P'{:}\set{\uparrow}(1,j_1), S'{:}\set{\downarrow}(1,j_2),
E'{:}\set{\oplus}(1,j_3),I'{:}\set{\ominus}(1,j_4)\}$. We only want
to compare species that appear in both sets and that have the same
role, hence we let $\NonInt = \{ P \}= \{P'\}$. We will show below that
the product is also the slow variable of both systems, illustrating
another way to determine $\NonInt$.

\begin{figure}
\begin{center}

\hspace*{-0.5cm}
{\renewcommand{\arraystretch}{0.85}
\begin{tabular}{c|c}
$Sys$: new variables & $Sys'$: new variables \\

$\begin{array}{rclclcll}
X_{S_T} & \!=\! & S + SE + P  & \!=\! & S_0 & \!=\! & n
&\text{conserved\ \ } \\
X_{E_T} & \!=\! & E + \EI + \SE & \!=\! & E_0 & \!=\! & m &\text{conserved} \\
X_{I_T} & \!=\! & \EI + I      & \!=\! & I_0 & \!=\! & p &\text{conserved} \\
X_{P} & \!=\! & P             &   &     & \!=\! & k &\text{slow} \\
X_{EI} & \!=\! & \EI           &   &     & \!=\! & l &\text{fast} \\
X_{SE} & \!=\! & \SE           &   &     & \!=\! & j &\text{fast} \\
\end{array}$

& 

$\begin{array}{rclclcll}
\ \ X_{S'_T} & \!=\! & S' + P' & \!=\! & S_0 & \!=\! & n &\text{conserved} \\
X_{E'} & \!=\! & E'_0      &   &     & \!=\! & m &\text{conserved} \\
X_{I'} & \!=\! & I'_0      &   &     & \!=\! & p &\text{conserved} \\
X_{P'} & \!=\! & P'        &   &     & \!=\! & k &\text{slow} \\
\\
\\
\end{array}$
\\

new state: $(P,\EI,\SE)$ & new state: $(P)$
\end{tabular}}
\end{center}
\caption{Identification of conserved, fast and slow variables}
\label{fig:variables}
\end{figure}

Next define the relation $\calR$ as 

{\renewcommand{\arraystretch}{0.30}
$\begin{array}{c}
\\
\hspace*{-0.8cm}
\bigl\{ ((n\mns (k\pls j),m\mns (j\pls l),p\mns l,k,l,j),
(n\mm k,m,p,k)) \mid 
  0 \!\leq\! k \!\leq\! n, 0 \!\leq\! j \!\leq\! \min\{m,n\mns k\}, 
  0 \!\leq\! l \!\leq\! p, j\mpl l \!\leq\! m \bigr\}. \\
\\
\end{array}$}

\noindent This captures the idea suggested by Figure~\ref{fig:tss} that states
with the same level of product are those that should be paired in
$\calR$. We now show that $\calR$ is a fast-slow bisimulation
for $\Af$ using the approach given in the previous section.
Figure~\ref{fig:variables} provides the new variables for each model.
Here, variables with subscript $0$ indicate initial values for those
species. First three invariants are identified, then we consider
just the fast transitions and this allows us to determine which
species are not affected by the fast transitions. $P$ is not affected
and neither is $S + \SE$. Since these are not linearly independent
(due to the first invariant), we need to choose one of them, and
we choose $P$ since it is a single species. There are no other
linearly independent slow variables so we need to find two fast
variables that are linearly independent from each other and the
four defined variables.  $\EI$ and $\SE$ are suitable candidates. The
technique can also be applied to the variables in $Sys'$ where there
are no fast variables since the only reaction $\gamma$ is slow.

Hence the states of the transition systems can be transformed
without changing the labels on the transitions. 
The new transition systems
have the same form as the original transition systems, but the new states
are vectors with the first three elements of the original vector
removed. A new relation can be defined
over these new transition systems that preserves the relationship between
states. Let
$\calR' = \bigl\{ ((k,l,j), (k)) \mid 
  0 \!\leq\! k \!\leq\! n, 0 \!\leq\! j \!\leq\! \min\{m,n\mns k\}, 
  0 \!\leq\! l \!\leq\! p, j\mpl l \!\leq\! m \bigr\}$.

Since $\calR'$ has the form required for the application of Proposition
\ref{prop:slowfast}, and the two models have the same slow variables,
if $\calR'$ is a slow bisimulation for $\As$, then it is a fast-slow
bisimulation for $\Af$.  The new transition system is isomorphic
to the original transition system and the relationship between
states is preserved by $\calR'$, hence $\calR$ is also a fast-slow
bisimulation for $\Af$ over the original transition system.

We now proceed with the proof that $\calR'$ is a slow bisimulation
for $\As$.  For notational convenience, we let $\wP_i = \{
P{:}\set{\uparrow}(1,i) \}$, and consider in turn the different cases for
which there are $\gamma$ transitions.

\begin{itemizeless}
\item Consider $((k,l,j),(k)) \in \calR'$ for $0\leq k<n$, $0 \leq
l \leq p$, $0<j\leq \min\{m,n\mm k\}$ which represent states where
some substrate-enzyme compound available. Then $(k)
\xyrightarrow{\gamma,\wP_k} (k\pls 1)$ is matched by $(k,l,j)
\xyrightarrow{\gamma,\wP_k} (k\pls 1,l,j\mns 1)$ and \emph{vice
versa}.

\item
Consider $((k,l,0),(k)) \in \calR'$ for $0 \leq k < n$, $0 \leq l
\leq p$ when no substrate-enzyme compound is present. There are
three cases depending on the relationship of $m$ and $p$.

\newpage

\begin{itemizeless}
\item
If $m>p$, consider $0 \leq l \leq p$.
Since $m$ is greater than $p$, whatever $l$ is, there will be
additional enzyme to form $SE$ and $(k) \xyrightarrow{\gamma,\wP_k}
(k\pls 1)$ is matched by $(k,l,0) \twoheadrightarrow (k,l,1)
\xyrightarrow{\gamma,\wP_k} (k\mpl 1,l,0)$.

\item
If $m \leq p$ and $0 \leq l \leq m\mm 1$, then enzyme is available
and the previous case applies.

\item 
If $m \leq p$ and $l=m$, then all
enzyme is bound in $\EI$. Then $(k)$ $\xyrightarrow{\gamma,\wP_k}$
$(k\pls 1)$ is matched by $(k,0,m)$ $\twoheadrightarrow$ $(k,0,m\mm
1)$ $\twoheadrightarrow$ $(k,1,m\mm 1)$ $\xyrightarrow{\gamma,\wP_k}$
$(k\mpl 1,0,m\mm 1)$.

\end{itemizeless}

\end{itemizeless}
An example of the first subcase is illustrated in Figure~\ref{fig:tss}
in the unmodified transition system.  Consider the pair of states
$((2,3,0,3,0,0),(2,3,0,3))\in \calR$. The $\gamma$-transition from
$(2,3,0,3)$ to $(1,3,0,4)$ is matched by a fast transition from
$(2,3,0,3,0,0)$ to $(1,2,0,3,0,1)$ and a $\gamma$-transition from
the latter to $(1,3,0,4,0,0)$ and $((1,3,0,4,0,0),(1,3,0,4))\in
\calR$.

To conclude, we have shown for $\{\alpha_1,\alpha_{-1},\beta_1,\beta_{-1}\}
\subseteq \Af $, $\gamma \in \As$ for the models $Sys$ and $Sys'$
that 
$(n,m,p,0,0,0) \thickapprox_{\!\Af} (n,m,p,0)$ for
all positive $n$, $m$ and $p$ which covers all major cases of
interest. Hence, we can conclude that the simpler model demonstrates
the same behaviour (at a semi-quantitative level) as the more complex
model when we abstract from fast reactions.  We can apply the
congruence result: if $P$ is a Bio-PEPA model with no fast reactions
in $\{\alpha_1,\alpha_{-1},\beta_1,\beta_{-1}\}$, then since $Sys
\thickapprox_{\!\Af} Sys'$, we know that $P \syncstar Sys
\thickapprox_{\!\Af} P \syncstar Sys'$. This allows us to build new
systems, and also to replace the larger state space of $Sys$ with
the smaller one of $Sys'$.

\section{Related and further work}\label{sec:relwork}

Various approaches to modelling biological systems using process
algebra have been proposed including $\kappa$-calculus \cite{DanoDL:04a},
$\pi$-calculus \cite{RegeRS:02,PriaPRSS:01a,BlosBCP:06a}, Beta-binders
\cite{PriaPQ:04a}, Bio-Ambients \cite{RegeRPSCS:04a} sCCP
\cite{BortBP:08a} and the
continuous $\pi$-calculus \cite{KwiaKS:08}.  Most of these approaches
use stochastic simulation as their analysis tool, and very few
approaches have considered the use of semantic equivalences.  Both weak
bisimulation and context bisimulation are shown to be congruences for the
\texttt{bio}-$\kappa$-calculus. Context bisimulation 
allows for the modelling of cell interaction \cite{LaneLT:08}.
Observational equivalence has been used to show that CCS specifications
of elements of lactose operon regulation have the same behaviour
as more detailed models \cite{PintPFMR:07}.  In an example of
biological modelling using hybrid systems, bisimulation is used to
quotient the state space with respect to a subset of variables as
a technique for state space reduction \cite{AntoAPPSM:04}. Bisimulation
has also been used in the comparison of ambient-style models and
membrane-style models \cite{CiobCA:08} and the comparison of a
term-rewriting calculus and a simple brane calculus \cite{BarbBMMT:08}.
Other equivalences have been defined for Bio-PEPA.
Compression bisimilarity is based on the idea that different
discretisations of a system should have the same behaviour assuming
sufficient levels \cite{GalpGH:10}.  
Strong and weak 
bisimulation parameterised
by functions have also been defined \cite{GalpG:10}
and their use demonstrated on a model with alternative pathways. 
Further work is to determine whether fast-slow bisimilarity can be
expressed as $g$-bisimilarity.

Although fast-slow bisimilarity is defined in the context of Bio-PEPA, it
is applicable to any formalism with the same style of stratification
of molecular counts or discretisation of concentrations, such as the 
Petri net-based modelling framework of Heiner \emph{et al}
\cite{HeinHGD:08}.

QSSA has also been applied to stochastic simulation
\cite{GillG:77} either to obtain approximate rates \cite{CaoCGP:05a}
or in the case of slow-scale stochastic simulation
\cite{CaoCGP:05a,CaoCGP:05b} to identity slow and fast species which
then leads to the introduction of a virtual fast process representing
the fast species where slow reactions are removed.

As mentioned earlier, QSSA is a time scale separation technique.
There are other variants such as tQSSA which consider the total
substrate (both free and bound) in deriving reduced equations and is
applicable when $S_0 + K_m \gg E_T$ does not hold
\cite{BorgBBS:96}. QSSA approaches have been formalised by single
perturbation theory \cite{SegeSS:89,ZagaZKK:04}.

Another form of time scale decomposition/separation considers CTMCs and is
based on a decomposition/ aggregation technique for solving for
steady state. In a nearly completely decomposable CTMC, the values
in the diagonal blocks are much larger (at least one order of magnitude)
than those in the off-diagonal
blocks \cite{CourC:77}. Hence there are blocks of states where transitions
between states in an block is much more frequent than transitions
between states in different blocks. The technique involves solving for
steady state for each block (ignoring transitions to other blocks). Each
block is then considered as a single state, and transition rates between
these states are computing, and the aggregate CTMC constructed is
solved. Finally, the solutions for each block and the aggregate CTMC
are combined to obtain an approximate solution for the original CTMC.

This technique has been applied to both stochastic Petri nets
\cite{BlakBT:93} and stochastic process algebra \cite{HillHM:95}.
For Petri nets, a function is defined over markings to determine
which markings are similar and must take into account relative
rates.  In the case of stochastic process algebra, an analysis of
processes and the rates of the actions they enable is the starting
point for identifying subsets of states. Sequential components are
categorised as fast, slow and hybrid, and states are grouped when
they have the same slow subcomponents. The passive rate can be used
to split hybrid processes into two sequential processes with the
same behaviour.  A time scale decomposition technique for transient
analysis \cite{BobbBT:86} is also relevant because our model considers
transient behaviour as well as steady state behaviour.  Since we
are not working fully quantitatively here, these approaches are
issues for further research. Specifically, we wish to compare the
application of the technique for nearly completely decomposable CTMCs with
a QSSA-based quantitative equivalence, as well as considering transient
behaviour.

Quantitative equivalences have been defined for CTMCs based on
Kripke structures, hence with unlabelled transitions and labelled
states \cite{BaieBKHW:05}. Both weak bisimulation and weak simulation
are defined. Further research involves applying these equivalences,
after suitable modification to CTMCs obtained from labelled transition
systems and seeing their relationship with the QSSA. 

Most previous CTMC research assumes fixed rates; however, with
Bio-PEPA rates are state-dependent which introduces additional
complexity.

\section{Conclusion}\label{sec:conclusion}

We have developed fast-slow bisimilarity, a semi-quantitative semantic
equivalence motivated by the Quasi-Steady-State Assumption. We show that
for two operators of interest, fast-slow bisimilarity is a congruence.
For the cooperation operator, a reasonable condition is required to ensure
congruence. The second operator is an
extension operator which allows a species to be extended with new reactive
capabilities. Although the definition of fast-slow bisimilarity is 
similar to that of weak bisimilarity, the condition for congruence for
cooperation illustrates how they differ. For certains types of reduced
models, it is possible to work with slow bisimilarity which only
considers slow reactions.  The use of fast-slow bisimilarity is
illustrated with an example of competitive inhibition, where one
system includes the intermediate compounds and the other does not.

This equivalence can be used to show that a reduced system has the
same behaviour as the full system. This means it is possible to
work only with the reduced system, thereby reducing the number of
parameters that need to be fitted.  Fast-slow bisimilarity can be
applied in any context where concentrations are discretised or
molecule counts are grouped.

Further work includes a fully quantitative equivalence, automation of
the bisimulation technique including variable reduction and
investigation of dynamically changing the sets of fast and slow
reactions.

\medskip
\newpage

\textbf{Acknowledgements:}
This work was supported by the EPSRC through Projects
EP/C54370X/01, EP/E031439/1 and EP/C543696/01. The Centre for
Systems Biology at Edinburgh (CSBE) is a Centre for Integrative Systems
Biology (CISB) funded by the BBSRC and EPSRC in 2006.


\vspace*{-0.01cm}

\bibliographystyle{eptcs}
\bibliography{CompModPaper4}

\end{document}